\documentclass[journal,twoside,web]{ieeecolor}

\usepackage{generic}
\usepackage{cite}
\usepackage{amsmath,amssymb,amsfonts}
\usepackage{algorithmic}
\usepackage{graphicx}
\usepackage{textcomp}

\usepackage{graphicx}
\usepackage{subfigure}
\usepackage{epsfig} 
\usepackage{cite}
\usepackage{lcsys}

\newtheorem{theorem}{Theorem}[section]

\newtheorem{remark}[theorem]{Remark}

\DeclareMathOperator*{\essinf}{ess\,inf}
\begin{document}

\def\BibTeX{{\rm B\kern-.05em{\sc i\kern-.025em b}\kern-.08em
    T\kern-.1667em\lower.7ex\hbox{E}\kern-.125emX}}
\markboth{\journalname, VOL. XX, NO. XX, XXXX 2017}
{Author \MakeLowercase{\textit{et al.}}: Preparation of Papers for IEEE Control Systems Letters (August 2022)}

\title{\LARGE \bf Passive iFIR filters for data-driven control
}

\author{Zixing Wang$^{a}$,  Yongkang Huo$^{a}$ and  Fulvio Forni$^{a}$
\thanks{*The work of Zixing Wang was supported by CSC Cambridge Scholarship.  The work of Yongkang Huo was supported by the UK Engineering and Physical Sciences Research Council (EPSRC) grant 10671447.}
\thanks{$^{a}$ Department of Engineering, University of Cambridge, CB2 1PZ Cambridge, U.K.
         \{{\tt\small zxw20@.cam.ac.uk}, {\tt\small yh415@cam.ac.uk}, {\tt\small f.forni@eng.cam.ac.uk}\}.}
}
\maketitle
\thispagestyle{empty}

\begin{abstract}
We consider the design of a new class of passive iFIR controllers given by the parallel action of an integrator and a finite impulse response filter. iFIRs are more expressive than PID controllers but retain their features and simplicity. The paper provides a model-free data-driven design for passive iFIR controllers based on virtual reference feedback tuning. Passivity is enforced through constrained optimization (three different formulations are discussed). The proposed design does not rely on large datasets or accurate plant models.
\end{abstract}


\section{Introduction}

\IEEEPARstart{F}{rom} simple prototyping to advanced control applications, a PID controller goes a long way due to its simple structure, intuitive tuning, and the presence of integral action for perfect regulation. In addition, with passive plants, PID control leads to inherently stable closed loops, even when the plant is uncertain and its model is approximated. {\color{black}This makes PID control one of the most common control techniques in industry 
\cite{ortega1998eulerbook, spong2022historical, ortega2023pid}}.

In many control applications, a stabilizing
feedback action must also satisfy complex performance requirements. 
These are often specified as desired input/output behaviors, 
in the form of a set of data or a desired transfer function. In this setting, the PID control shows limitations \textcolor{black}{\cite{BCE2011,Yonezawa&Kajuwara2023-Direct_Frictious_reference-stability}}. 
Ideally, we want {\color{black} a controller that retains the simplicity of PID control and its fundamental properties, namely integral action and passivity, but is more flexible.}

In this paper, we investigate the design of iFIR controllers, which combine the parallel action of a finite impulse response filter (FIR) and an integrator for perfect regulation. Similar formulations, also based on cascade interconnections, can be found in \cite{Formentin2011,vanHeusden2011,Yahagi2022-overfitting,deJong2023}.
The idea is to replace the proportional and derivative action of the PID controller with the richer action of a FIR filter. We retain the integrator for perfect regulation. We also restrict the FIR filter to be passive, to guarantee wide applicability in the control of electro-mechanical systems and robotics \cite{ortega1998eulerbook, spong2022historical, ortega2023pid}. {\color{black}With these features}, our hypothesis is that iFIR controllers provide a more flexible alternative to PID control, when combined with data-driven optimal tuning.

We propose a model-free data-driven design approach based on
virtual reference feedback tuning (VRFT) \cite{Campi2002}, which reduces the matching problem with a target reference model to a least-squares fitting. No plant model is needed. Passivity is enforced through constrained optimization. 
The novelty of the paper is in endowing virtual reference feedback tuning with additional constraints for passivity. These are based on the Kalman-Yakubovich-Popov (KYP) lemma, on the Toeplitz operator of the iFIR controller, and on the positive realness of its transfer function. The result is a design approach based on convex constrained optimization, that can be solved with standard tools, such as CVX \cite{cvx2,cvx1,cvxpy}. In what follows, we show that passivity constraints based on Toeplitz and positive realness formulations scale better with the size of the filter, strongly improving the complexity of the KYP lemma.

For passive plants, passive iFIR control design separates closed-loop stability and closed-loop performance. Data scarcity and low-quality data do not affect the stability of the closed loop, which is structurally guaranteed via passivity. The quality and quantity of data matter for the performance, to achieve good matching with the desired reference model. This separation between stability and performance is novel. To guarantee stability, other data-driven approaches require \textcolor{black}{the probing signal to be persistently exciting \cite{BCE2011} or informative \cite{wang&Henk}}, or provide asymptotic guarantees in the limit of infinite data length \cite{vanHeusden2011, Yonezawa&Kajuwara2023-Direct_Frictious_reference-stability, Selvi2021-VRFT-stability}.

The idea of enforcing passivity combined with a data-driven approximation is inspired by the kernel approach to nonlinear modeling offered by \cite{bib:vanWaarde2023}. There, passivity is achieved through regularization and the scattering transform. We achieve it through constrained optimization, in the narrower setting of iFIR architectures. 
\textcolor{black}{Passive filter design has 
been also explored in \cite{prakash2022} and \cite{Wu_Boyd1996}.
These approaches show strong similarities with our formulation. Our contribution is in the combination with VRFT for control design purposes (Theorems \ref{thm:KYP}, \ref{thm:finitepptLMI}, and \ref{thm:sample_f}), in the formulation based on Toeplitz matrices (Theorems \ref{thm:pptLMI} and \ref{thm:finitepptLMI}), and in the characterization of exact bounds on (sampled) frequency-based constraints (Theorem \ref{thm:sample_f}), which lead to a non-iterative design.}

The paper is structured as follows. The iFIR control design based on VRFT is illustrated in Section 2. Passivity constraints are extensively discussed in Section 3.  Section 4 provides an example of closed-loop control with linear and nonlinear plants. The conclusions follow.

\section{Data-driven design of iFIR controllers}
\label{sec:VRFT}

iFIR controllers $C$ combine integration and finite impulse response filtering. Given the finite coefficients $\{g_k\}_{k=0,...,m-1}$,
the controller is given by
\begin{equation}  \label{eq:sec2-9-d}
    C(z) \ = \ \underbrace{\frac{\gamma T_{s}}{1-z^{-1}}}_{\mbox{integrator}} + \underbrace{\ \sum_{k=0}^{m-1} g_k  z^{-k}}_{\mbox{FIR filter}}
\end{equation}
where the integrator is discretized using backward Euler discretization with sampling period $T_s$. $\gamma$ is the integral gain.
In fact, $C$ shows a generalized PID structure, where the proportional and derivative terms are replaced by a FIR filter.

$C$ can be easily designed from data, following the approach of virtual reference feedback tuning \cite{Campi2002}.
In Figure~\ref{fig:sec2-1}, we represent the desired closed-loop performance with a
time-invariant (possibly nonlinear) reference model $M_r$. $P$ is a time-invariant single-input single-output plant.
The plant is not necessarily linear but we assume \emph{passivity} from the input $u$ to the output $y$ \textcolor{black}{(e.g. Euler-Lagrangian systems and port-Hamiltonian systems \cite{ortega1998eulerbook,secchiPH2007}). For any given} reference $r$,
the ideal closed-loop behavior satisfies $y \simeq y^*$. 
\vspace{-1mm} 

\begin{figure}[htbp]
    \begin{center}
        \includegraphics[width=.62\columnwidth]{ 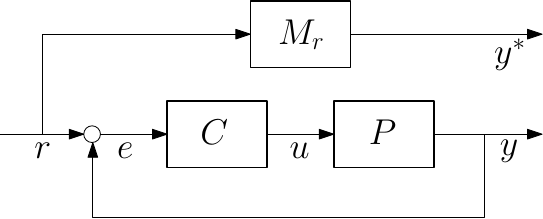}   
        \caption{Block diagram of the closed-loop system.} 
        \label{fig:sec2-1}
    \end{center} \vspace{-3mm}
\end{figure}

The controller can be derived as the solution of the following optimization problem:
\begin{equation} \label{eq:vrft_obj}
    \min_{C\in \mathrm{iFIR}} \frac{1}{N} \sum_{t=0}^{N-1} \left( u(t) - C ( M_{r}^{-1}(y) - y)(t)  \right)^{2}
\end{equation}
where $\{y(t)\}_{t=0,...,N-1}$ is the measured open-loop plant output when the plant is excited with a user-defined signal $\{u(t)\}_{t=0,...,N-1}$.
The output signal $y$ can be obtained by experiments, making the approach suitable for settings where a plant model is not fully available.

Under the ideal condition of perfect matching $y=y^*$, the idea of VRFT is to
derive a compatible reference signal $r = M_r^{-1} y^* = M_r^{-1} y$, which is then used to
obtain the controller $C$ by fitting the input/output pair $(r-y,u)$. 
The inverse reference model $M_r^{-1}$ is realized by a causal approximation.
For linear reference models, the inverse can be removed by filtering of input and output signals via $M_r$ itself, 
which leads to the filtered optimization problem 
\begin{equation} \label{eq:vrft_obj_filter}
    \min_{C\in \mathrm{iFIR}} \frac{1}{N} \sum_{t=0}^{N-1} \left( M_{r}u(t) - C ( y - M_{r}y)(t)  \right)^{2}
\end{equation}

Following \cite{Campi2002}, for linear plants and linear reference models, the minimizer of \eqref{eq:vrft_obj}  also minimizes
\begin{equation}\label{eq:mr_obj}
    \min_{C(z) } \left \| M_r(z) - \frac{P(z)C(z)}{1+P(z)C(z)} \right\|_{2}^{2}
\end{equation}
if the optimal controller of the latter, $C^{*}(z)$, can be represented by an iFIR \eqref{eq:sec2-9-d} (for $N \to \infty$ and if $u$ is given by the realization of a stationary and ergodic stochastic process).
If $C^{*} \notin {\mathrm{iFIR}}$, pre-filtering of the open-loop data $\{u(t),y(t)\}_{t=0,...,N-1}$ with a suitably chosen 
LTI filter can reduce the difference between the two minimizers.
We refer to \cite{Campi2002} for additional details on VRFT.

    The optimization problems in \eqref{eq:vrft_obj} and \eqref{eq:vrft_obj_filter} are equivalent to a least-squares problem and hence can be solved efficiently. For instance, take $e = M_{r}^{-1}(y) - y$. From \cite[Eq. 7.3]{pillonetto2022BOOK}, the optimization problem \eqref{eq:vrft_obj} (or \eqref{eq:vrft_obj_filter}) with the controller \eqref{eq:sec2-9-d} reads
    \begin{equation} \label{eq:vrft_obj_ls}
        \min_{g\in \mathbb{R}^{m}} \left \|  u - Eg \textcolor{black}{-} \gamma E_{\mathrm{int}}       \right\|_{2}^{2}
    \end{equation}
    where 
    $\small{
    g =\begin{bmatrix}
        g_0  & \dots & g_{m-1} \\
    \end{bmatrix}^T}
    $,
    $\small{
    u = \begin{bmatrix}
        u(0) & \dots & u(N-1) \\
    \end{bmatrix}^T}
    $, 
    $$
        \small{E =
        \begin{bmatrix}
        e(0) & 0 & \dots & \dots  & 0 \\
        e(1) & e(0) & 0 & \dots & 0 \\
        \vdots & \vdots &  &  &  \vdots\\
        e(N-2) & e(N-3) & \dots & \dots & e(N-m-1) \\
        e(N-1) & e(N-2) & \dots & \dots & e(N-m)
    \end{bmatrix}}, 
    $$
    and
    $\small{E_{\mathrm{int}} = T_{s} \begin{bmatrix}
        e(0) & e(0)+e(1) & \dots & \sum_{t=0}^{N-1} e(t)
    \end{bmatrix}^{T}.}$ \textcolor{black}{If the input is persistently exciting, then $E$ will be full rank and a unique solution exists \cite[Section 3.2]{BCE2011}. 
    In what follows, in our extension of VRFT based on passive iFIR controller, uniqueness of solution is not necessary to guarantee closed-loop stability. Therefore, this rank condition will not be considered again.}

\section{Passive iFIR controllers}
\label{sec:passive_iFIR}
\textcolor{black}{Although VRFT is an established and effective design methodology, the data-driven optimal controller \eqref{eq:vrft_obj} \emph{does not guarantee closed-loop stability}. For passive plants $P$, we tackle this issue through passivity.
We combine the least-square fitting \eqref{eq:vrft_obj_ls} with specific constraints that enforce passivity.} This leads to the derivation of the passive iFIR controller that best fits the desired input / output data \eqref{eq:vrft_obj}.

The first formulation is a straightforward application of the Kalman-Yakubovich-Popov (KYP) lemma. 
\begin{theorem}[KYP approach]
\label{thm:KYP}
    Given the iFIR \eqref{eq:sec2-9-d} of order $m$, consider the state-space realization of the FIR part 
    \begin{equation}
    	\begin{cases} \label{eq:sec2-13}
    		\quad x(t+1) = \textcolor{black}{A_{c}}x(t) + \textcolor{black}{B_{c}}\textcolor{black}{w(t)}\\
    		\quad \textcolor{black}{z(t)} = \textcolor{black}{C_{c}}\textcolor{black}{x(t)} + \textcolor{black}{D_{c}}\textcolor{black}{w(t)}
    	\end{cases}
    \end{equation}
    where \textcolor{black}{$w(t),z(t)$ are the input and output of the FIR part.} 
    \textcolor{black}{$A_{c}\in \mathbb{R}^{m-1 \times m-1},B_{c}  \in \mathbb{R}^{m-1} ,C_{c} \in \mathbb{R}^{m-1},D_{c} \in \mathbb{R}$} 
    are given by
    \begin{align}
        \textcolor{black}{A_{c}} &= \begin{bmatrix}
            0 & 1 & 0 &0 &\dots &0 \\
            0 & 0 & 1 &0 &\dots &0 \\
            \vdots & & & & & \\
            0 & 0 & 0 &0 &\dots &1 \\
            0 & 0 & 0 &\dots & &0 
        \end{bmatrix} , \quad  
        \textcolor{black}{B_{c}} = \begin{bmatrix}
            0\\
            0  \\
            \vdots \\
            0 \\
            1
        \end{bmatrix} \nonumber \\
        \textcolor{black}{C_{c}} &= 
        \begin{bmatrix}
            g_{m-1} \quad g_{m-2} \quad \dots \quad g_{2} \quad g_{1}
        \end{bmatrix}, \quad 
        \textcolor{black}{D_{c}} = g_{0} \, . \nonumber 
    \end{align}
   \textcolor{black}{The optimal iFIR controller given by \eqref{eq:vrft_obj} 
   constrained to the following LMIs in the unknown 
   $X\in \mathbb{R}^{m-1\times m-1}$, $C_c$, and $D_c$
    \begin{subequations} \label{eq:sec2-15}
        \begin{align}
        	   \gamma & \geq  0 \\
                X=X^T & >  0  \label{eq:sec2-15-1}\\
                \quad \begin{bmatrix}
                X - A_{c}^{T}XA_{c} & C_{c}^{T} - A_{c}^{T}XB_{c} \\
                C_{c}-B_{c}^{T}XA_{c}  & D_{c} + D_{c}^{T} - B_{c}^{T}XB_{c}
                \end{bmatrix} &\geq  0 \, \\
                D_{c} + D_{c}^{T} &\geq 0. \label{eq:sec2-15-2}
        \end{align}
    \end{subequations}
    is passive}, thus guarantees closed-loop stability.
\end{theorem}
\begin{proof}
   \eqref{eq:sec2-15} are the standard matrix inequalities for passivity of discrete LTI systems. These are linear in the \textcolor{black}{unknowns $C_c$ and $X$.} 
   The passivity of the iFIR controller follows from the KYP lemma \textcolor{black}{\cite[Thm. 1]{Wu_Boyd1996}}, recalling that the parallel interconnection of passive systems is passive. 
   Closed-loop stability follows from the passivity theorem.  
\end{proof}

The attractive feature of Theorem \ref{thm:KYP} is that the constrained minimization problem can be handled efficiently
by convex optimization solvers such as CVX, \cite{cvx2,cvx1,cvxpy}.
The main issue of \eqref{eq:sec2-15} is that the number of unknowns scales as $o(m^2)$.
This may cause long computation times and numerical precision issues, leading to convergence issues when
the order of the FIR part is high. The approach proposed below, overcomes these limitations.

Define the Toeplitz matrix of order $n \geq m$ associated with the coefficients $g = \{g_{k}\}_{k=0,...,m-1}$ as follows, \cite{gray2006}:
\begin{equation}  \label{eq:sec2-3}
    \phi_{n}(g) = 
    \begin{bmatrix}
        g_0 & 0 & \dots & \dots & 0 \\
        g_1 & g_0 & 0 & \dots & 0 \\
        \vdots & \vdots & \ddots & \vdots & \vdots \\
        g_{m-1} & g_{m-2} & \dots & 0 & 0 \\
        0 & g_{m-1} & \dots & g_0 & 0 \\
        0 & \dots & g_{m-1} & \dots & g_0 \\
    \end{bmatrix} \in \mathbb{R}^{n \times n} .
\end{equation}

The first $n$ samples of the response of the FIR filter given by $g$ can be computed 
as the product between $\phi_n(g)$ and the (vectorized) input 
\textcolor{black}{\small{$w_{n}= \begin{bmatrix}
        w(0) & w(1) & \dots & w(n-1)
    \end{bmatrix}^{T}.$}}
It is thus not surprising that the passivity of the FIR filter is related \textcolor{black}{to} the positive semi-definiteness of 
the Toeplitz matrix $\phi_{n}(g) + \phi_{n}(g)^{T}$, as stated in the following theorem. 
In what follows, the symbol $\underline \sigma$ denotes the minimal singular value of a matrix.

\begin{theorem}[Infinite Toeplitz approach] \label{thm:pptLMI}
    Given the iFIR \eqref{eq:sec2-9-d} of order $m$, let $g = \{g_{k}\}_{k=0,...,m-1}$ be the finite coefficients of its FIR part.
   The iFIR controller given by \eqref{eq:vrft_obj} and by
   \begin{subequations}     
   \label{eq:pptLMI} 
\begin{align}
\gamma & \geq 0 \\
\lim_{n\to\infty} \underline \sigma \! \left( \phi_n^T(g) + \phi_n(g) \right) & \geq 0  \label{eq:infinite_toplitz_phi}
\end{align}
\end{subequations}
     is passive (thus guarantees closed-loop stability).
\end{theorem}
\begin{proof}
    Like the proof of Theorem \ref{thm:KYP}, we just need to show that the FIR part of the filter is passive.
    Consider \textcolor{black}{any} vectorized input \textcolor{black}{\small{$w_{n}= \begin{bmatrix}
        w(0) & w(1) & \dots & w(n-1)
    \end{bmatrix}^{T}$}} and associated output \textcolor{black}{\small{$z_{n}= \begin{bmatrix}
        z(0) & z(1) & \dots & z(n-1)
    \end{bmatrix}^{T}$}} of the FIR part, where $n \geq m$. \textcolor{black}{Denote by $w_{\infty},z_{\infty}$ the related (infinite length) vectors when $n \to \infty$. For passivity, we require $z_{\infty}^{T}w_{\infty} \geq 0$ for all $w_{\infty}$.} For all $n \geq m$, we have
    \begin{equation*}
        z_n^T w_n = w_{n}^{T} \phi_{n}(g)^T w_{n} = w_{n}^{T} \left( \frac{\phi_{n}(g)^T + \phi_{n}(g)}{2} \right ) w_{n} \ .
    \end{equation*}
    \textcolor{black}{Therefore, saying that 
    $z_n^T w_n \geq 0$ for all inputs $w_n$ (and related outputs $z_n$) is equivalent to the matrix inequality
    $$\phi_{n}(g)^T + \phi_{n}(g) \geq 0.$$ 
    Thus, for $n \to \infty$, \eqref{eq:infinite_toplitz_phi} is equivalent to $z_{\infty}^{T}w_{\infty} \geq 0$.}
\end{proof}

The main issue of Theorem \ref{thm:pptLMI} is that \eqref{eq:pptLMI} is not computationally tractable. 
This motivates the following relaxation. 
\begin{theorem}[Finite Toeplitz approach] \label{thm:finitepptLMI}
    Given the iFIR \eqref{eq:sec2-9-d} of order $m$, let $g = \{g_{k}\}_{k=0,...,m-1}$ be the finite coefficients of its FIR part. 
   For any $\epsilon>0$, $\rho_0 > 0$, $0 < \rho \leq 1$, 
   there exists $n^* \geq m$ such that,
   for all $n \geq n^*$, 
    the iFIR controller given by \eqref{eq:vrft_obj} and 
    \begin{subequations}     
     \label{eq:toeplitz_constr}
	\begin{align}
	\gamma &\geq 0 \label{eq:finite_toeplitz_gamma}\\
	\rho_0 &\geq  |g_{0}|  \label{eq:finite_toeplitz_g0}\\
	  \rho_0 \rho^k &\geq |g_k|  \quad \forall k \in \{ 1,...,m-1\} \label{eq:finite_toeplitz_gk}\\
	\phi_n^T(g) + \phi_n(g) & \geq \epsilon I_n \label{eq:finite_toeplitz_phi}
	\end{align}
    \end{subequations}
    is passive.
\end{theorem}
\begin{proof}
    We show that the FIR part of the controller is passive. 
For $n\to \infty$, the matrix $\phi_{n}(g)^T + \phi_{n}(g)$ is a $(m-1)-banded$ symmetric infinite dimensional Toeplitz matrix, \cite[Thm. 2.1]{strohmer2002},  \cite[p. 37]{gray2006}. Hence, as shown in 
   \cite[Eq. 4.5]{gray2006} and \cite[p. 322]{strohmer2002},
   there exists an associated generating function $f$ given by
    \begin{equation} \label{eq:f}
        f(\omega) = g_0 + \sum_{k=-m+1}^{m-1}g_{k}e^{j2\pi k \omega}
    \end{equation}
    where $\omega \in [-\frac{1}{2},\frac{1}{2})$.
    $\{g_{k}\}_{k=-\infty,...,\infty}$ is absolutely summable.
    Therefore, we can apply \cite[Corollary 4.2]{gray2006} to obtain
    \begin{equation}
        \lim_{n\to\infty} \underline \sigma \! \left(\phi_{n}(g)^T + \phi_{n}(g)\right) = \essinf_{\omega} f(\omega) \leq \infty
    \end{equation}
    where $\essinf$ refers to essential infimum defined in \cite[p. 40]{gray2006}. Since $|f(\omega)| \leq |g_0| + \sum_{k=-m+1}^{m-1} |g_k|$,
    it follows that the limit 
    is bounded.

   From the definition of limit, 
   it follows that for any given $\rho_0$, $\rho$, and $\varepsilon$ satisfying the assumptions of the theorem, there exists $n^*$ such that 
    \begin{equation} \label{eq:convergence}
    \left| \lim_{k\to\infty}\underline \sigma \! \left( \phi_k^T(g) + \phi_k(g) \right) - \underline \sigma \! \left( \phi_n^T(g) + \phi_n(g) \right) \right| \leq \epsilon
    \end{equation}
    for all $n \geq n^*$. Thus, for any given $n\geq n^*$, \eqref{eq:finite_toeplitz_phi} guarantees 
    $\lim_{k\to\infty}\underline \sigma \! \left( \phi_k^T(g) + \phi_k(g) \right) \geq 0$.
    This implies that the FIR part of the iFIR controller is passive.
\end{proof}

$\rho_0$, and $\rho$ are not strictly needed for the proof of the theorem.
However, these user-defined parameters determine the decay rate of the FIR coefficients, and can be used to reduce the Gibbs phenomenon due to
sharp truncation of the impulse response of the FIR part of the controller.
\eqref{eq:finite_toeplitz_phi} is a relaxation
of \eqref{eq:infinite_toplitz_phi} based on
a finite constraint plus an additional bound $\epsilon$.
The latter reduces as $n$ increases (and vice versa). 
Numerical tests show that, for large $n$, the Nyquist diagram of a FIR filter that satisfies $\phi_{n}(g)^T + \phi_{n}(g) \geq 0$ has a large portion of its locus on the right-half of the complex plane. 
To have the locus completely within the right-half plane, we can `shift' the whole Nyquist diagram right using strictly input passivity, \cite{chaffey2023graphical}. 
This is equivalent to taking $ \phi_{n}(g)^T + \phi_{n}(g) \geq \epsilon I$.

Theorem \ref{thm:finitepptLMI} provides a computationally tractable optimization problem. A passive iFIR can be obtained by solving the optimization iteratively, for increasing values of $n$:
(i) we select $\rho_0$, $\rho$ and $\epsilon$ and solve the optimization problem for a given $n\geq m$;
(ii) we test passivity of the iFIR (from the Nyquist plot or using any other test);
(iii) in case of shortage of passivity, we select a larger $n$ (and possibly a larger $\epsilon$) and iterate.

Compared with the KYP approach of Theorem \ref{thm:KYP}, the finite Toeplitz approach of Theorem \ref{thm:pptLMI} 
leads to an optimization problem with $o(m)$ unknowns, linear in the order $m$ of the iFIR filter. \textcolor{black}{However, for a fixed $\epsilon$ (or $n$), the value of $n$ (or $\epsilon$) that leads to passivity is bounded by some finite but unknown  $n^*$ (by some finite $\epsilon^*$), thus
needs to be tuned iteratively.} 
The following theorem overcomes this limitation. Given the coefficients $g={g_k}$, in what follows we use
$G(e^{j\theta})$ to denote its frequency response $\sum_{k=0}^{m-1}g_{k}e^{-jk\theta}$.

\begin{theorem}[Positive realness approach] \label{thm:sample_f}
    Given the iFIR \eqref{eq:sec2-9-d} of order $m$, let $g = \{g_{k}\}_{k=0,...,m-1}$ be the finite coefficients of its FIR part. For any $\rho_{0} > 0$, $0<\rho \leq 1$, $M \geq 2$ the iFIR controller by \eqref{eq:vrft_obj} and 
    \begin{subequations}  
        \label{eq:spr}
    	\begin{align}
        	\gamma &\geq 0 \\
        	\rho_0 &\geq  |g_{0}| \label{eq:samplef_g0}\\
        	  \rho_0 \rho^k &\geq |g_k|  \quad \forall k \in \{ 1,...,m-1\}  \label{eq:samplef_gk}\\
           G\left( e^{ \frac{j q }{M}\pi } \right) + G\left(e^{- \frac{j q }{M}\pi }\right) &\geq \epsilon \quad \forall q \in \{ 0,..., M\}  \label{eq:samplef_G}\\
           \epsilon &= \pi\rho_{0}\frac{1-\rho^{m}}{1-\rho} \frac{m-1}{2M} \label{eq:samplef_epsl}
    	\end{align}
    \end{subequations}
    is passive. 
\end{theorem}
\begin{proof}    
    We focus on the FIR part of the filter, to show that \eqref{eq:spr}
    implies positive realness, i.e. $G(e^{j\theta}) + G(e^{-j\theta}) \geq 0$ 
    for all $\theta$. The latter is equivalent to 
    \begin{equation}
        \sum_{k=0}^{m-1}g_{k}\cos(k\theta) \geq 0 \quad \forall \theta \in [0,\pi]. \label{eq:lc}
    \end{equation}

    As first step, we need the following bound.
    Take $f(\theta) = \sum_{k=0}^{m-1}g_{k}\cos(k\theta)$. 
    For all $\theta$ and $\Delta$, we have
    \begin{align}
        |f(\theta + \Delta) - f(\theta)|  
        & \leq \sum_{k=0}^{m-1} |g_{k}| \left| \cos(k\theta + k\Delta) - \cos(k\theta) \right|  \nonumber\\
        & \leq \sum_{k=0}^{m-1} |g_{k}| \left|k \Delta \right| \nonumber\\
        & \leq (m-1)|\Delta| \sum_{k=0}^{m-1} |g_{k}|.  \label{eq:samplef_1}
    \end{align}
    
    We now use \eqref{eq:spr}. Let us sample $f$ at $M + 1$ points over the domain $[0,\pi]$ with a sampling interval $\frac{1}{M}$. That is, consider $f\left(\frac{q}{M}\pi\right)$ for $q\in \{0,..., M\}$. Then, for any given $\theta \in [0,\pi]$, find the closest sample 
    point $\frac{q}{M}\pi$ to $\theta$ and take $\Delta = \theta - \frac{q}{M}\pi $.
    By construction, $-\frac{\pi}{2M} \leq \Delta \leq \frac{\pi}{2M}$ and
    \eqref{eq:samplef_1} guarantees that 
    \begin{align}
        \left|f(\theta) - f\left(\frac{q}{M}\pi\right)\right|
         &= \left|f\left(\frac{q}{M}\pi  + \Delta\right) - f\left(\frac{q}{M}\pi\right)\right|  \nonumber \\
         &\leq (m-1) |\Delta| \sum_{k=0}^{m-1} |g_{k}| \nonumber \\
         &\leq \frac{m-1}{2M} \pi \sum_{k=0}^{m-1} |g_{k}| \nonumber \\
         &\leq \frac{m-1}{2M} \pi \rho_{0}\frac{1-\rho^m}{1-\rho}, \label{eq:samplef_2}
    \end{align}
    where the last inequality follows from \eqref{eq:samplef_g0} and \eqref{eq:samplef_gk}, 
    since  $ \sum_{k=0}^{m-1} |g_{k}| \leq \sum_{k=0}^{m-1}\rho_{0}\rho^{k}$.

    Thus, from \eqref{eq:samplef_2}, it is easy to see that 
    \eqref{eq:samplef_G} and \eqref{eq:samplef_epsl}
    guarantee $G(e^{j \theta} ) + G(e^{-j \theta} ) = f(\theta) \geq 0$.
\end{proof}

Theorem \ref{thm:sample_f} follows closely \cite[Section III.B]{prakash2022} 
\textcolor{black}{and \cite[Section 2]{Wu_Boyd1996} but \emph{guarantees} passivity with a finite length of training data and a finite number of constraints in the context of data-driven control design. Compared with Theorem \ref{thm:finitepptLMI}, Theorem \ref{thm:sample_f} leads to a one-shot design approach.}
The theorem illustrates that passivity can be enforced by constraining the sampled frequency response of the FIR to be larger than a certain constant $\epsilon$, which depends on the decay rate of the impulse response, the order of the FIR part, and the number of sampling points in the frequency domain. This is equivalent to verifying that the sampled Nyquist diagram of the FIR subcontroller has a real part greater than $\epsilon$. In this case, $\epsilon$ bounds the maximal variation of the real part of the Nyquist diagram between each two samples.

\textcolor{black}{From the proof of Theorem \ref{thm:sample_f}, it can be seen that the bound \eqref{eq:samplef_epsl} on $\epsilon$ is conservative. Conservativeness is
reduced by increasing the number of samples $M$, that is, 
by enforcing a larger number of constraints. 
}
The trade-off between the computational complexity and conservativeness can be mitigated by replacing the
\textcolor{black}{bound} \eqref{eq:samplef_epsl} with an heuristic search over $\epsilon$, based on an iteration similar to the one derived from Theorem \ref{thm:finitepptLMI}.
In terms of complexity, Theorem \ref{thm:sample_f} is more efficient than Theorems \ref{thm:KYP} and \ref{thm:finitepptLMI} since there are no LMIs involved. Furthermore, \eqref{eq:samplef_G} can be implemented with $M + 1$ linear inequality constraints of the form \eqref{eq:lc}, for $\theta = \frac{q \pi }{M}$
and $q \in \{0, \dots, M\}$.

\begin{remark} \label{remark_existence_solution}
    \textcolor{black}{The constrained optimization \eqref{eq:vrft_obj} subject to passivity constraints is always feasible. For instance, for any $\epsilon^{*} \in \mathbb{R}$ and $\gamma \geq 0$, consider the iFIR controller in \eqref{eq:sec2-9-d} given by $g_{0} \geq \epsilon^{*}$ and $g_{k} = 0, \forall k > 0$.
    If $\epsilon^{*} \geq 0$, this iFIR controller is compatible with \eqref{eq:sec2-15}. If $\epsilon^{*}$ is larger than the $\epsilon$ chosen in \eqref{eq:toeplitz_constr} or \eqref{eq:samplef_epsl}, then this iFIR controller is compatible with \eqref{eq:toeplitz_constr} or \eqref{eq:spr}, respectively.}
\end{remark}
\begin{remark}
    \textcolor{black}{The phase correction that a passive linear controller can introduce is limited to $\pm 90$ degrees. This poses a limit on the achievable performance. The level of conservativeness depends on the performance requirement, based on the selected reference model $M_r$. }
\end{remark}
\begin{figure}[htbp]
\vspace{-2mm}
    \begin{center}
        \includegraphics[width=0.44\columnwidth]{ 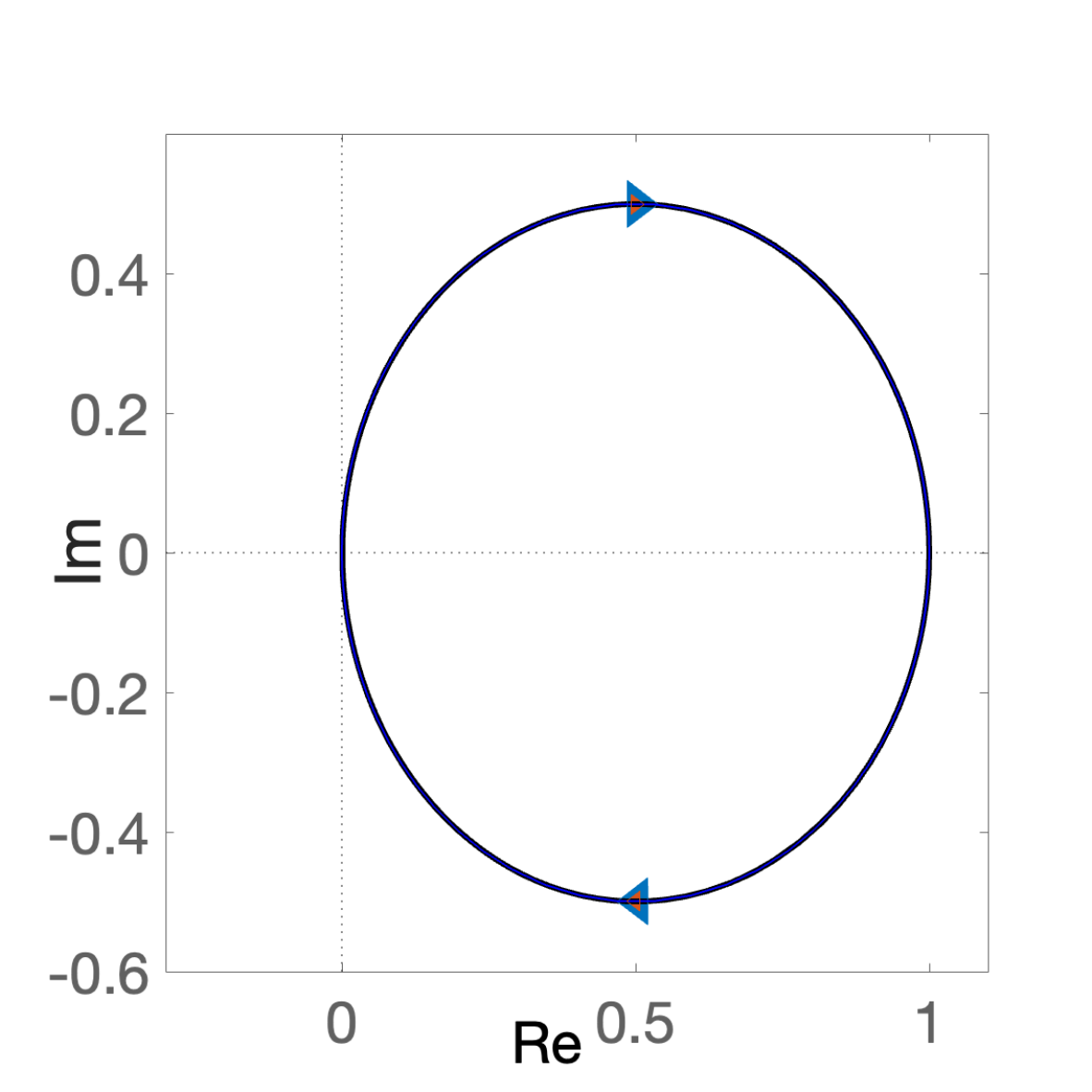}
        \includegraphics[width=0.44\columnwidth]{ 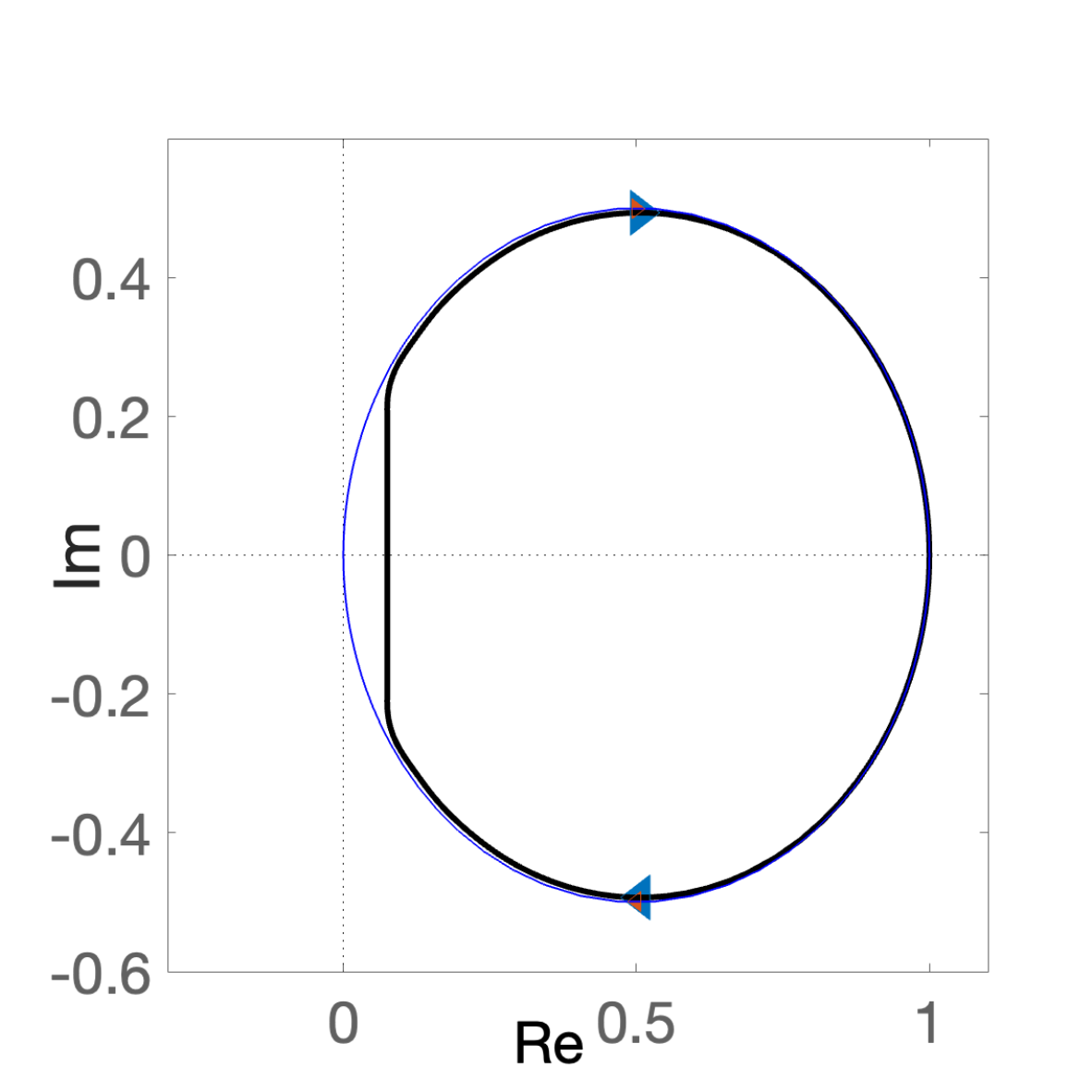}    
        \includegraphics[width=0.44\columnwidth]{ 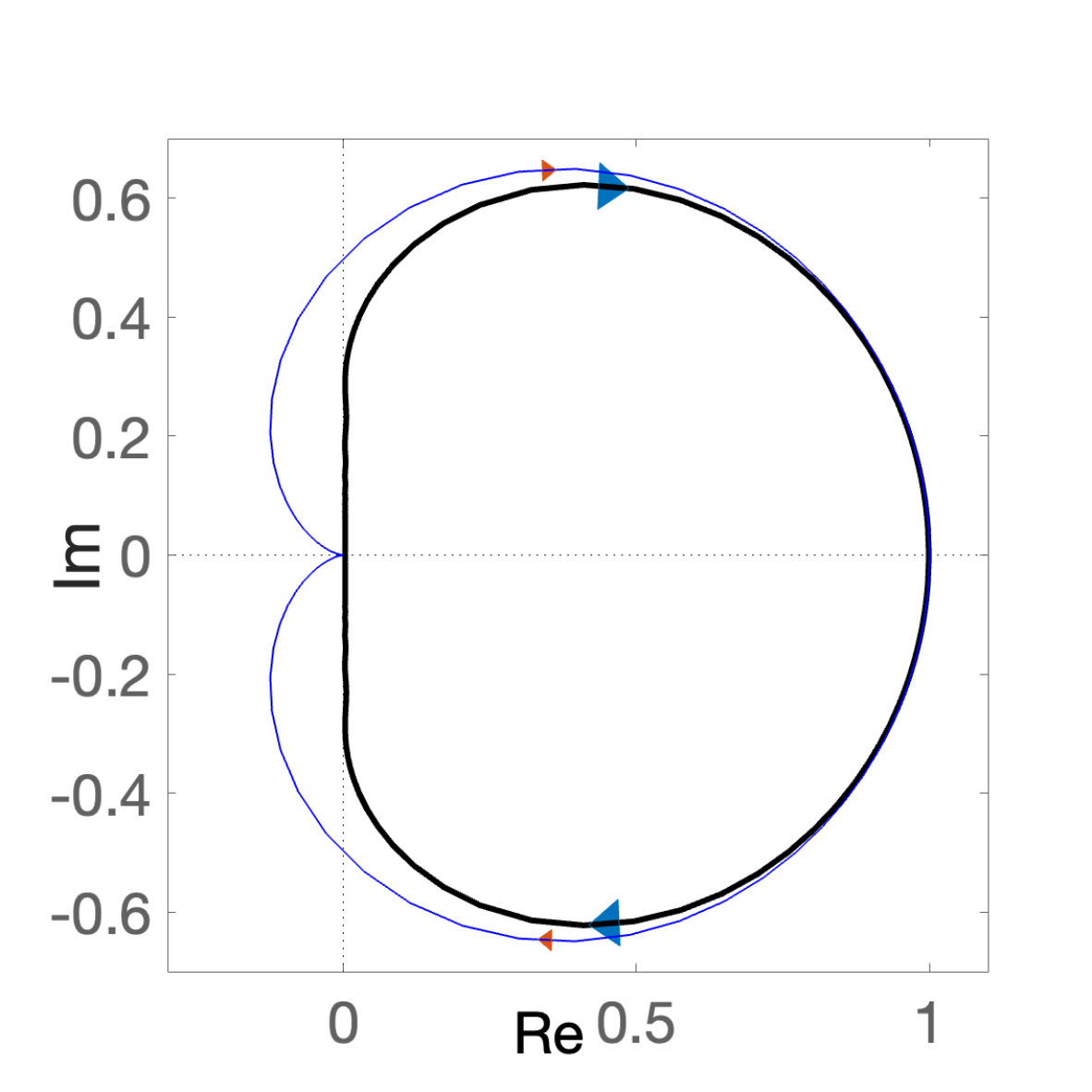}
        \includegraphics[width=0.44\columnwidth]{ 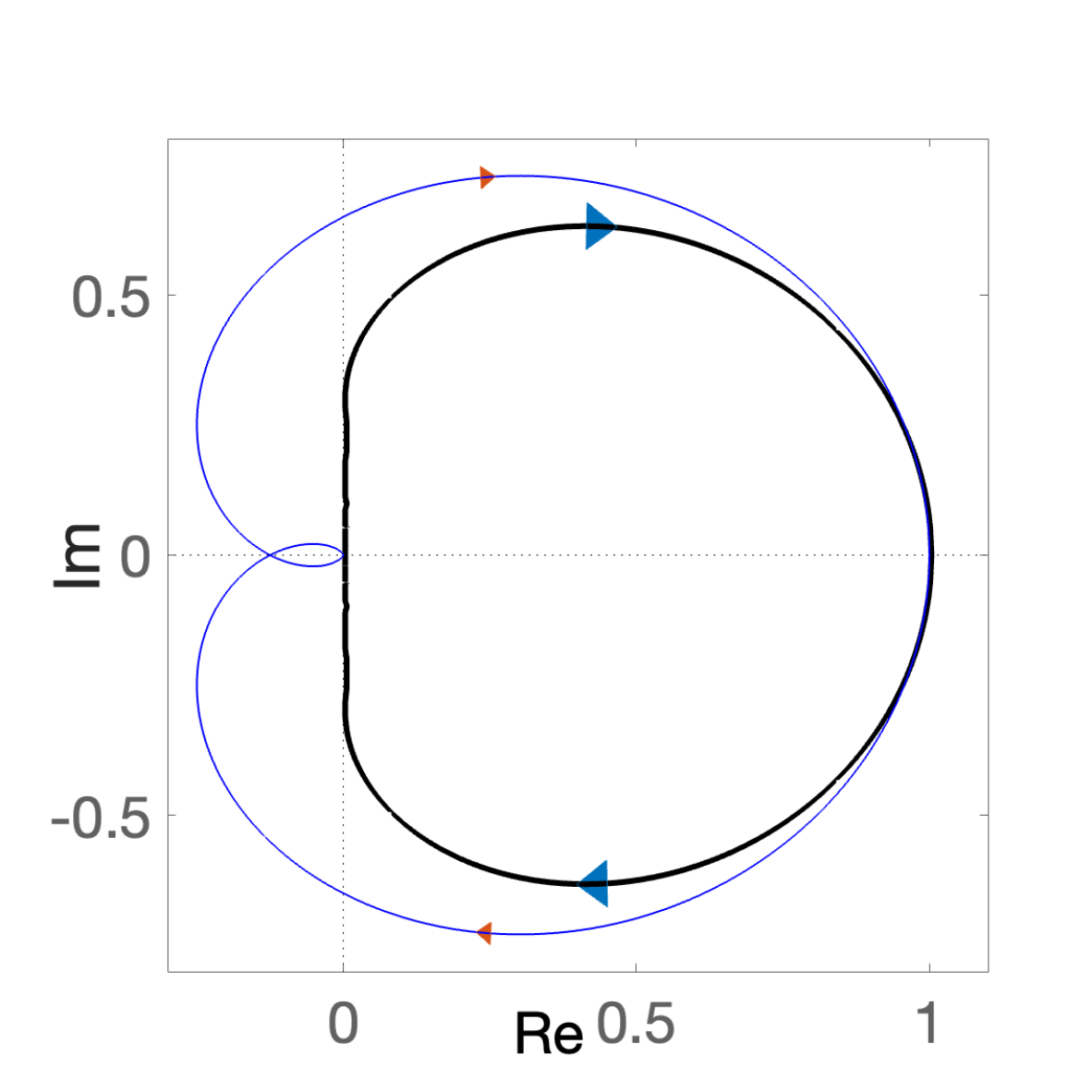}
        \vspace{-2mm}
        \caption{Nyquist diagrams of Example \ref{sec:example1}. 
        \textcolor{black}{B}lue - target filter; black - identified passive FIR filter.
        Top: $C_1$ (left), $C_1$ and excess of passivity (right).
        Bottom: $C_2$ (left), $C_3$ (right).}
        \label{fig:nyquist_sysid}
    \end{center} \vspace{-7mm}
\end{figure}

\section{Examples}
\label{sec:example}
\subsection{Fitting of passive FIR filters}
\label{sec:example1}
We identify passive FIR filters to illustrate
the effectiveness of \eqref{eq:sec2-15}, \eqref{eq:toeplitz_constr} and \eqref{eq:spr} 
in dealing with the targets
\begin{equation}
C_{q}(s) = \frac{1}{(0.5s+1)^q}  \qquad q \in\{1,2,3\},
\end{equation}
where $C_{1}$ is passive, and $C_2$ and $C_3$ show increase shortage of passivity.
For each target, we take $\gamma=0$ and  \eqref{eq:vrft_obj_ls} 
paired with passivity constraints.
\textcolor{black}{The input/output data $(e,u)$ are generated as $u = C_ie$
for $i \in \{1,2,3\}$.  
The probing signal is a filtered step given by $e = \mathcal{L}^{-1}\left(\frac{1}{0.2s+1}\frac{1}{s}\right)$, where $\mathcal{L}^{-1}$ is the inverse Laplace transform.
}
Results are shown in Figure~\ref{fig:nyquist_sysid}. The top-left plot
shows that all three approaches give perfect fitting. The top-right
plot shows the results for conservative $\epsilon$ \eqref{eq:finite_toeplitz_phi}.
Indeed, this parameter can be used to enforce an excess of passivity,
to compensate for shortages in the plant. The left and right bottom plots are related to the non-passive targets $C_{2}$ and $C_{3}$. 
All three approaches achieve a passive FIR filter. As expected,
the match with the target systems degrades as the passivity shortage increases. \textcolor{black}{These results are obtained using CVXPY \cite{cvxpy} with MOSEK solver \cite{mosek}.} 
\subsection{Velocity control of a compliant mechanical system}
\label{sec:example2}
\textcolor{black}
{Consider the two-cart plant of Figure~\ref{fig:two_carts}, modeling a truck pulling a heavy load. 
The cart masses are $m_{1} =  3.0 \mbox{ kg}, m_{2} = 0.5 \mbox{ kg}$. $k_{12}= 1.0$ Ns/m is the spring stiffness and $c_{12} = 1.05$ Ns/m is the damping coefficient.} Each cart is also affected by additional dissipation of the form $c \dot{x}_i$, where $c=0.5$ Ns/m and $\dot{x}_i$ is the absolute velocity of the cart $i \in \{1,2\}$. 
The block diagram for this closed-loop system is represented in Figure~\ref{fig:sec2-1}. $u$ is the applied force and $y$ is the velocity output. 
\begin{figure}[htbp]
\vspace{-2mm}
    \begin{center}
        \includegraphics[width=0.6\columnwidth]{ 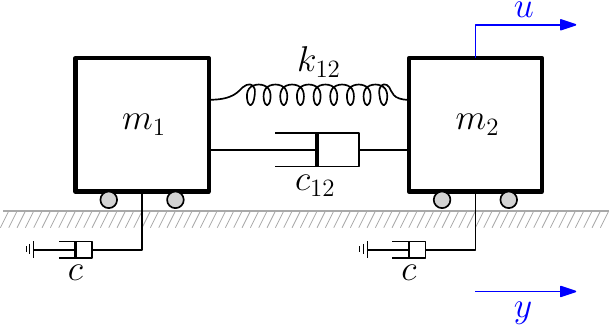}  
        \vspace{-2mm}  
        \caption{A compliant two-cart system.} 
        \label{fig:two_carts} 
    \end{center}  
\end{figure}
\vspace{-2mm}

We consider 
\textcolor{black}{the two reference models $M_{r1}$ and  $M_{r2}$}
\begin{align*}
    \textcolor{black}{M_{ri}(s)} &\textcolor{black}{= \frac{1+0.1s}{1+2\zeta sT_{i} + s^2T_{i}^2}
    \qquad i = 1,2}
\end{align*}
\textcolor{black}{where $\zeta = 1$, $T_{1} = 0.25$s, $T_{2} = 0.125$s. $M_{r2}$ is more aggressive because it has a higher cutoff frequency.} The plant and the reference model\textcolor{black}{s} are discretized
with \textcolor{black}{the zero-order-hold method at} a sampling time \textcolor{black}{$T_{s}$} of $0.05$s. The plant is excited in open-loop with \textcolor{black}{$u(t) = \sum_{i=1}^{10} \sin(\omega_{i}T_{s}t)$} where $\omega_{i}$ linearly spans the frequency range of $[0.5,10]$ rad/s, and \textcolor{black}{$t \in \{0,1,...,2000\}$}. \textcolor{black}{We assume that there is no noise in the input and output of the plant. Standard techniques such as instrumental variables \cite[Section 4]{Campi2002} can be used in the noisy case.}
The data $(e,u)$ for the controller are obtained as discussed in Section \ref{sec:VRFT}
and filtered with the filter\textcolor{black}{s $M_{r1}(z)$ and $M_{r2}(z)$.} 

\textcolor{black}{For $M_{r1}$}, the computation times for several iFIR controllers of order $m\in\{50,150,250,350\}$ are summarized in Table \ref{tb:speed}. \textcolor{black}{The times refer to the solution of the optimization problem on a standard Apple Silicon M2. Compilation times are excluded.} For $m=350$, the KYP approach takes more than one hour. This shows the advantage of 
\eqref{eq:toeplitz_constr} and \eqref{eq:spr} over \eqref{eq:sec2-15}. \vspace{-2mm}

\begin{table}[htbp]
    \begin{center}
        \begin{tabular}{c|cccc} 
                 &  m=50 & m=150 & m=250 & m=350   \\ \hline
            \textcolor{black}{Thm. \ref{thm:KYP}} & \textcolor{black}{0.24} s& \textcolor{black}{54.03} s & \textcolor{black}{705.44}  s & $/$  \\ \hline
            \textcolor{black}{Thm. \ref{thm:finitepptLMI}}, $n=m$ & 0.07 s& 0.51 s & 2.81 s &4.88 s  \\ \hline
            \textcolor{black}{Thm. \ref{thm:finitepptLMI}}, $n=1.5m$ & 0.12 s& 1.88 s & 5.21 s & 12.18 s  \\ \hline
            \textcolor{black}{Thm. \ref{thm:finitepptLMI}}, $n=2m$ & 0.17 s & 3.36 s & 10.62 s & 61.14 s  \\ 
            \hline
            \textcolor{black}{Thm. \ref{thm:sample_f}}, $M=m$ & 0.05 s& 0.14 s & 0.17 s & 0.25 s  \\ \hline
            \textcolor{black}{Thm. \ref{thm:sample_f}}, $M=1.5m$ & 0.04 s& 0.14 s & 0.17 s & 0.24 s  \\ \hline
            \textcolor{black}{Thm. \ref{thm:sample_f}}, $M=2m$ & 0.05 s & 0.14 s & 0.20 s & 0.26 s  \\ 
        \end{tabular} 
         \caption{\textcolor{black}{The computation times of passive iFIR design.}} \label{tb:speed} \vspace{-6mm}
    \end{center}
\end{table}

The Nyquist plots of the FIR part of the iFIR controllers 
are shown in Figure~\ref{fig:sec21-2}. \textcolor{black}{Since plots are similar, only the Toeplitz case is shown for reasons of space.} We have taken $m=350$, $n=2m$ for \eqref{eq:toeplitz_constr} and $M=2m$ for \eqref{eq:spr}. $\epsilon$ is chosen heuristically. 
On the right-part of the complex plane, both passive iFIR controllers (black) show excellent matching with the optimal unconstrained ones (blue) .

\begin{figure}[htbp]
    \begin{center}
        \includegraphics[width=0.49\columnwidth]{ 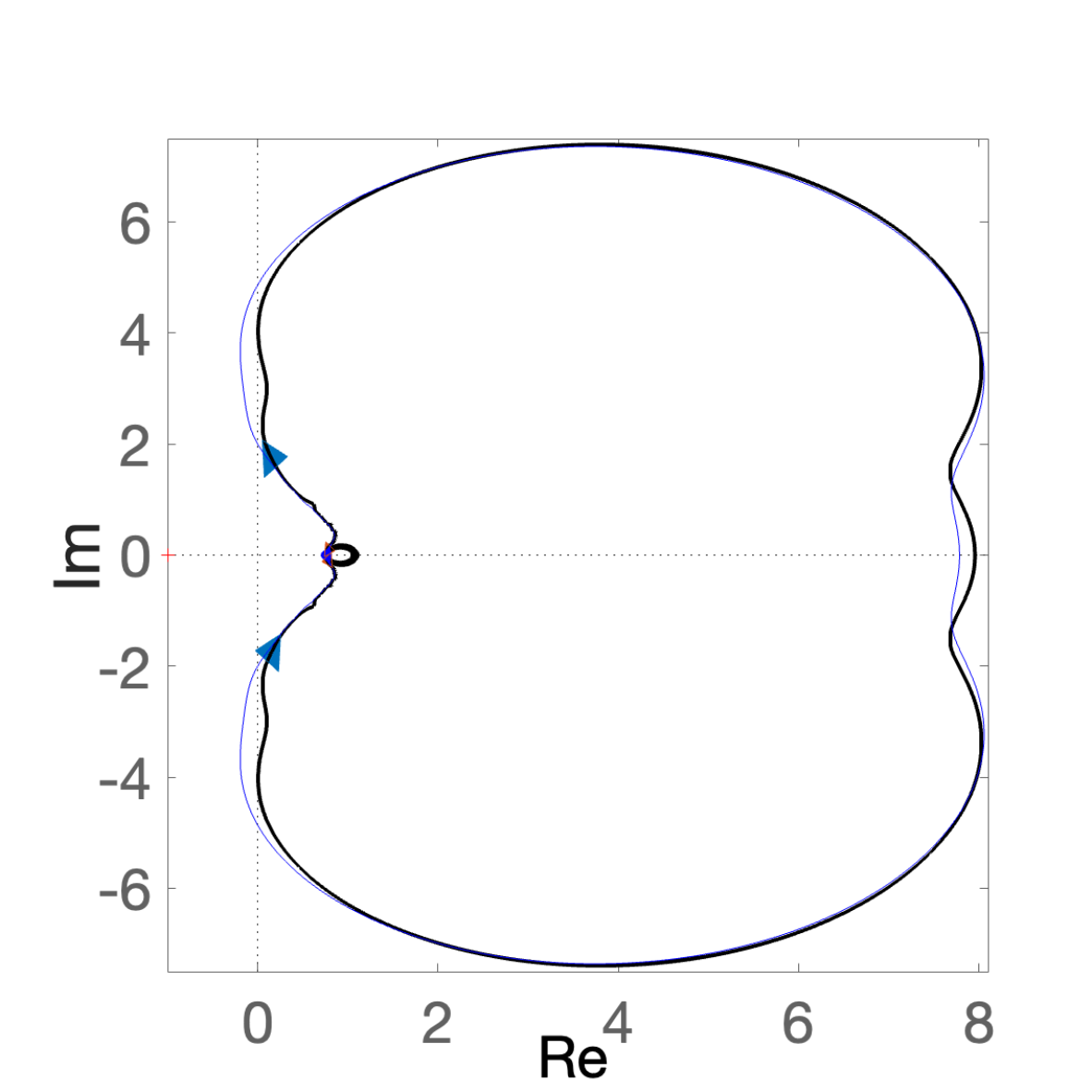}
        \includegraphics[width=0.49\columnwidth]{ 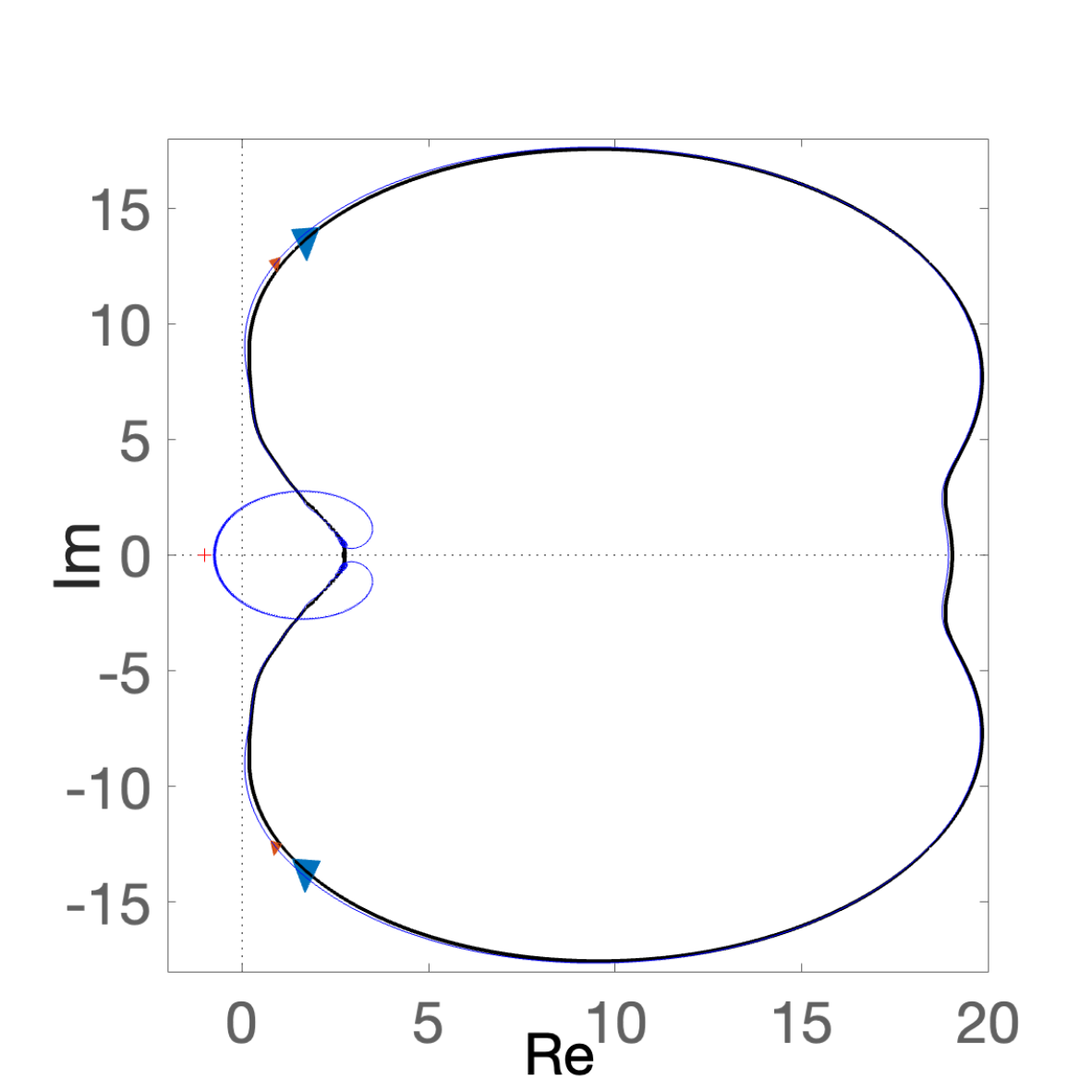}  
        \vspace{-4mm}
        \caption{Nyquist plots of the FIR part of the iFIR controller. \textcolor{black}{Blue - unconstrained \eqref{eq:vrft_obj}; black - constrained \eqref{eq:vrft_obj}, \eqref{eq:toeplitz_constr}. 
        Left: $M_{r1}$; Right: $M_{r2}$.}
        } 
        \label{fig:sec21-2} 
    \end{center} 
\end{figure}

The iFIR controllers are tested in simulation, as shown in Figure~\ref{fig:sec21-3}.
The figure shows a comparison between passive iFIR controllers and \textcolor{black}{the optimal PID controller 
$C(z) = K_{p} +K_{d}\frac{T_{s}}{1-z^{-1}} + K_{i}\frac{z-1}{zT_{s}} $ where $[K_{p},K_{d},K_{i}]$ correspond to $[0.8051,4.4090,0.0068]$ and $[2.6142,10.2330,0.0232]$ for $M_{r1}$ and $M_{r2}$, respectively. The PID gains are constrained to be non-negative (for passivity) and tuned using VRFT.
For space reasons, we do not develop a comparison with non-passive iFIR controllers and PID controllers}. 

\begin{figure}[htbp]
    \begin{center}
        \includegraphics[width=0.49\columnwidth]{ 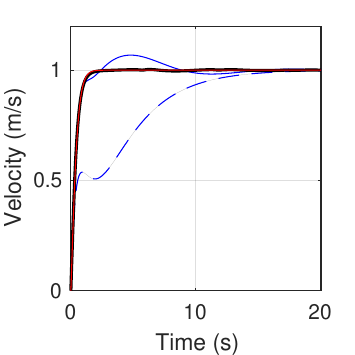}
        \includegraphics[width=0.49\columnwidth]{ 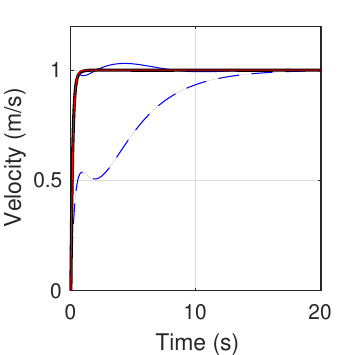} 
        \vspace{-6mm}
        \caption{Step responses. Red: target response. Black: closed-loop response with \textcolor{black}{passive} iFIR controller \textcolor{black}{ \eqref{eq:vrft_obj}, \eqref{eq:toeplitz_constr}.}
        Blue-continuous: closed-loop response with PID. Blue-dashed: open loop. \textcolor{black}{Left: $M_{r1}$; Right: $M_{r2}$.} 
        }  
        \label{fig:sec21-3}
    \end{center} 
\end{figure}

The closed-loop response of the two iFIR controllers is superior, with
an almost exact matching to the target model. 
This is further illustrated by the Bode diagrams of the closed-loop transfer functions, in Figure~\ref{fig:bode}.

\begin{figure}[htbp]
\vspace{-6mm}
    \begin{center}
        \includegraphics[width=0.49\columnwidth]{ 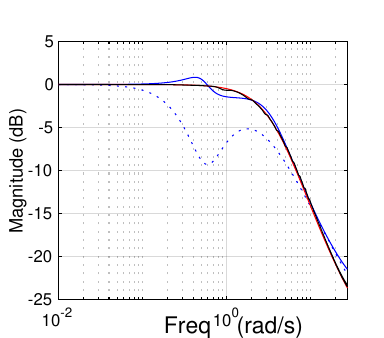} 
        \includegraphics[width=0.49\columnwidth]{ 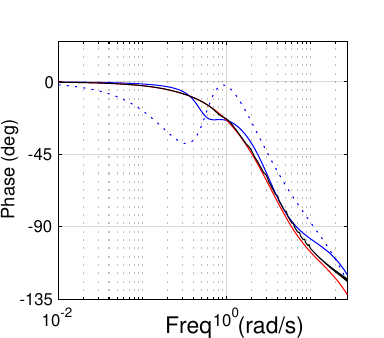} 
        \caption{Bode diagrams \textcolor{black}{when $M_{r1}$ is used. Left: magnitude plot. Right: phase plot.} 
        Red: target reference model. 
        Black:  iFIR-based \eqref{eq:vrft_obj},
        \textcolor{black}{\eqref{eq:toeplitz_constr}} closed loop.
        Blue-continuous: PID based closed loop.
        \textcolor{black}{Blue-dotted:} plant. \textcolor{black}{Similar results hold for $M_{r2}$.}}
        \label{fig:bode}
    \end{center} \vspace{-8mm}
\end{figure}

We also design and test the iFIR controller on a nonlinear passive plant. 
We replace the linear spring between the carts with a piece-wise linear one. The force-displacement characteristic has
slope $1$ for small displacements ($|\Delta x| \leq 0.5$) and $2$ for large displacements ($|\Delta x| \geq 0.5$).
To deal with the nonlinearity, we reduce the aggressiveness of the 
reference model, by taking 
\textcolor{black}{ the reference model to be $\frac{1+0.1s}{1+2\zeta sT + s^2T^2},\zeta = 1$, $T =1$s}. The iFIR controller is designed using \eqref{eq:vrft_obj} and \eqref{eq:spr}, based on nonlinear plant data. 
Performance degrades but stability is guaranteed. 
The iFIR design  outperforms the optimal PID,
as shown in Figure \ref{fig:time_nl}.

\vspace{-3mm} 
\begin{figure}[htbp]
    \begin{center}
        \includegraphics[width=0.7\columnwidth]{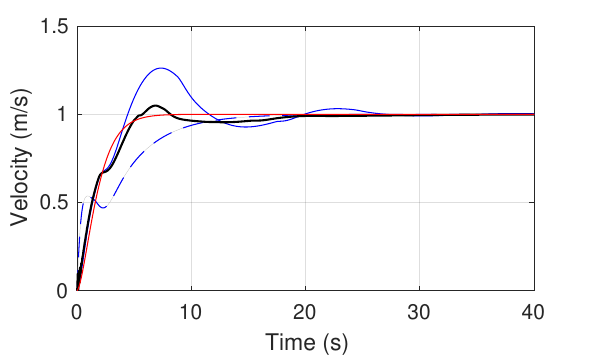}
        \includegraphics[width=0.22\columnwidth]{ 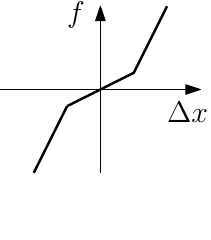} \vspace{-2mm} 
        \caption{Left: step responses of the controlled nonlinear plant. 
        Red: target response. 
        Black: nonlinear closed-loop response with iFIR controller. 
        Blue-continuous: nonlinear closed-loop response with PID. 
        Blue-dashed: nonlinear open loop. Right: nonlinear spring force-displacement characteristics.} 
        \label{fig:time_nl}
    \end{center} \vspace{-6mm}
\end{figure}

\section{Conclusions}

We have discussed a new data-driven approach for passive iFIR controllers.
Their design is based on convex optimization, combining virtual reference feedback tuning with passivity constraints. Our approach does not require any plant model or large datasets, and the plant can be nonlinear. For passive plants, the stability of the closed loop is guaranteed by construction. 
\textcolor{black}{Future work will extend the design to MIMO systems and to nonlinear controllers. }

\bibliographystyle{IEEEtran}
\bibliography{IEEEabrv,bib,science}

\end{document}